\documentclass{article}
\usepackage[utf8]{inputenc}
\usepackage{fullpage}

\usepackage{bm}
\usepackage{amsthm,amsmath,amssymb}
\usepackage{thmtools}
\usepackage{thm-restate}

\usepackage{hyperref}

\usepackage{cleveref}

\newtheorem{theorem}{Theorem}[section]

\newtheorem{lemma}[theorem]{Lemma}
\newtheorem{claim}[theorem]{Claim}

\theoremstyle{definition}
\newtheorem{definition}{Definition}[section]

\theoremstyle{remark}

\newcommand{\E}{\mathbb{E}}
\newcommand{\satisfying}[1]{\ensuremath{\mathsf{sol}(#1)}}
\newcommand{\est}{\ensuremath{\mathsf{Est}}}
\newcommand{\algo}{\ensuremath{\mathsf{Alg}}}
\newcommand{\formula}{\ensuremath{\phi}}
\newcommand{\queryformula}{\ensuremath{\psi}}
\newcommand{\tf}{\ensuremath{\{T, F\}}}
\newcommand{\vars}{\ensuremath{\mathsf{vars}}}
\newcommand{\At}{\ensuremath{\mathsf{At}}}
\newcommand{\good}{\ensuremath{\mathsf{Good}}}
\newcommand{\sat}{\ensuremath{\mathsf{SAT\textrm{-}Sample}}}
\newcommand{\num}[2]{\ensuremath{\mathsf{N}_{#2}(#1)}}
\newcommand{\counter}{\ensuremath{\mathsf{SAT\textrm{-}Sample}\ $counter$ }}
\newcommand{\counters}{\ensuremath{\mathsf{SAT\textrm{-}Sample}\ $counters$ }}

\DeclareMathOperator*{\argmin}{arg\,min}

\title{Approximate Model Counting: Is SAT Oracle More Powerful than NP Oracle?}

\author{Diptarka Chakraborty\footnote{National University of Singapore, Singapore. Supported in part by an MoE AcRF Tier 2 grant (MOE-T2EP20221-0009) and Google South \& South-East Asia Research Award. Email: diptarka@comp.nus.edu.sg}\and Sourav Chakraborty \footnote{Indian Statistical Institute, Kolkata. Email: sourav@isical.ac.in} \and Gunjan Kumar \footnote{National University of Singapore, Singapore. Supported in part by  National Research Foundation Singapore under its NRF Fellowship Programme[NRF-NRFFAI1-2019-0004 ]. Email: dcsgunj@nus.edu.sg} \and Kuldeep S. Meel\footnote{National University of Singapore, Singapore. Supported in part by  National Research Foundation Singapore under its NRF Fellowship Programme[NRF-NRFFAI1-2019-0004 ]  and Campus for Research
Excellence and Technological Enterprise (CREATE) programme, Ministry of Education Singapore Tier 2 grant MOE-T2EP20121-0011, and Ministry of Education Singapore Tier 1 Grant [R-252-000-B59-114 ]. Email: meel@comp.nus.edu.sg}}

\begin{document}

\maketitle


\begin{abstract}

Given a Boolean formula $\formula$ over $n$ variables, the problem of model counting is to compute the number of solutions of $\formula$. Model counting is a fundamental problem in computer science with wide-ranging applications in domains such as quantified information leakage, probabilistic reasoning, network reliability, neural network verification, and more. Owing to the \#P-hardness of the problems, Stockmeyer initiated the study of the complexity of approximate counting. Stockmeyer showed that $\log n$  calls to an NP oracle are necessary and sufficient to achieve $(\varepsilon,\delta)$ guarantees. The hashing-based framework proposed by Stockmeyer has been very influential in designing practical counters over the past decade, wherein the SAT solver substitutes the NP oracle calls in practice. It is well known that an NP oracle does not fully capture the behavior of SAT solvers, as SAT solvers are also designed to provide satisfying assignments when a formula is satisfiable, without additional overhead. Accordingly, the notion of SAT oracle has been proposed to capture the behavior of SAT solver wherein given a Boolean formula, an SAT oracle returns a satisfying assignment if the formula is satisfiable or returns unsatisfiable otherwise. Since the practical state-of-the-art approximate counting techniques use SAT solvers, a natural question is whether an SAT oracle is more powerful than an NP oracle in the context of approximate model counting. 

The primary contribution of this work is to study the relative power of the NP oracle and SAT oracle in the context of approximate model counting. The previous techniques proposed in the context of an NP oracle are weak to provide strong bounds in the context of SAT oracle since, in contrast to an NP oracle that provides only one bit of information, a SAT oracle can provide $n$ bits of information. We therefore develop a new methodology to achieve the main result: a SAT oracle is no more powerful than an NP oracle in the context of approximate model counting. 

\end{abstract}

\section{Introduction}

Let $\formula$ be a Boolean formula over $n$ propositional variables. An assignment $s \in \tf^n$ is called a {\em satisfying assignment} if it makes $\formula$ evaluate to true. Let $\satisfying{\formula}$ denote the set of all satisfying assignments. 
The model counting problem is to compute $|\satisfying{\formula}|$ for a given $\formula$. 
It is a fundamental problem in computer science and has numerous applications across different fields such as quantified information leakage, probabilistic reasoning,  network reliability, neural network verification, and the like~\cite{roth1996hardness,sang2005performing,valiant1979complexity,fredrikson2014satisfiability,duenas2017counting,baluta2019quantitative}.
The seminal work of Valiant~\cite{valiant1979complexity} showed that the problem of model counting is \#P-complete, and consequently, one is often interested in approximate variants of the problem. In this paper, we consider the following problem:

\ 

\noindent{\bfseries Approximate Model Counting} 

\begin{description}
	\item[Input] A formula $\formula$, tolerance parameter $\varepsilon>0$, and confidence parameter $\delta \in (0,1)$. 
	\item[Output] Compute an estimate $\est$ such that 
	\begin{align*}
		\Pr\left[\frac{|\satisfying{\formula}|}{1+\epsilon} \le \est \le (1+\epsilon) |\satisfying{\formula}|\right] \ge 1- \delta.
	\end{align*}
\end{description}

Stockmeyer~\cite{stockmeyer1983complexity} initiated the study of the complexity of approximate model counting. Stockmeyer's seminal paper made two foundational contributions: the first contribution was to define the query model that could capture possible natural algorithms yet amenable enough to theoretical tools to allow non-trivial insight. To this end, Stockmeyer proposed the query model wherein one can construct an arbitrary set $S$ and query an NP oracle to determine if $|\satisfying{\formula} \cap S| \geq 1$. Stockmeyer showed that under the above-mentioned query model, $\log n$ calls to an NP oracle are necessary and sufficient (for a fixed $\varepsilon$ and $\delta$). Furthermore,  Stockmeyer introduced a hashing-based algorithmic procedure to achieve the desired upper bound that makes $O(\log n)$ calls to NP-oracle. The lack of availability of powerful reasoning systems for problems in NP dissuaded the development of algorithmic frameworks based on Stockmeyer's hashing-based framework until the early 2000s~\cite{gomes2006model}. 

Motivated by the availability of powerful SAT solvers, there has been a renaissance in the development of hashing-based algorithmic frameworks for model counting, wherein a call to an NP oracle is handled by an SAT solver in practice. The current state-of-the-art approximate model counter, ApproxMC~\cite{chakraborty2016algorithmic}, relies on the hashing-based framework and is able to routinely handle problems involving hundreds of thousands of variables. The past decade has witnessed a sustained interest in further enhancing the scalability of these approximate model counters. It is perhaps worth highlighting that Stockmeyer's query model captures queries by ApproxMC. 

While the current state-of-the-art approximate model counters rely on the hashing-based framework, they differ significantly from Stockmeyer's algorithm for approximate model counting. The departures from Stockmeyer's algorithm have been deliberate and have often been crucial to attaining scalability. In particular, ApproxMC crucially exploits the underlying SAT solver's ability to return a satisfying assignment to attain scalability. In this context, it is worth highlighting that, unlike an NP oracle that only returns the answer Yes or No for a given Boolean formula, all the known SAT solvers are capable of returning a satisfying assignment if the formula is satisfiable without incurring any additional overhead. Observe that one would need $n$ calls to an NP oracle to determine a satisfying assignment. From this viewpoint, an NP oracle does not fully capture the behavior of an SAT solver, and one needs a different notion to model the behavior of SAT solver. 

Delannoy and Meel~\cite{delannoy2022almost} sought to bridge the gap between theory and practice by proposing the notion of a SAT oracle. Formally, a SAT oracle takes in a Boolean formula $\formula$ as input and returns a satisfying assignment $s \in \satisfying{\formula}$ if $\formula$ is satisfiable and $\bot$, otherwise. It is worth highlighting that we may need $n$ calls to an NP oracle to simulate a query to a SAT oracle, and therefore, it is conceivable for an algorithm to make $O(\log n)$ calls to a SAT oracle but $O(n \log n)$ calls to an NP oracle. Delannoy and Meel showcased precisely such behavior in the context of \emph{almost-uniform generation}. Their proposed algorithm, $\mathsf{UniSamp}$ makes $O(\log n)$ calls to a SAT oracle and would require $O(n \log n)$ calls to an NP oracle if one were to replace a SAT oracle with an NP oracle. At the same time, it is not necessary that there would be a gap of $n$ calls for every algorithm: simply consider the problem of determining whether a formula is satisfiable or not. Only one call to an NP oracle (and similarly to a SAT oracle) suffices.

Furthermore, the notion of the SAT oracle has the potential to be a powerful tool to explain the behavior of algorithms, as highlighted by Delannoy and Meel. Given access to an NP oracle, the sampling algorithm due to Jerrum, Valiant, and Vazirani~\cite{jerrum1986random} (referred to as $\mathsf{JVV}$ algorithm) makes $O(n^2 \log n)$ calls to an NP oracle as well as a SAT oracle, i.e., there are no savings from the availability of a SAT oracle. On the other hand, the algorithm, $\mathsf{UniSamp}$ makes 
$O(\log n)$ and $O(n \log n)$  calls to SAT and an NP oracle respectively. Therefore, the NP oracle model would indicate that one should expect the performance gap between $\mathsf{JVV}$ and $\mathsf{UniSamp}$ to be linear, while the SAT oracle model indicates an exponential gap. The practical implementations of $\mathsf{JVV}$  and $\mathsf{UniSamp}$ indeed indicate the performance gap between them to be exponential rather than linear. Therefore, analyzing problems under the SAT oracle model has the promise to have wide-ranging consequences. 

In this paper, we analyze the complexity of the problem of approximate model counting given access to a SAT oracle. Our study is motivated by two observations: 

\begin{description}

	\item[O1]  The modern state-of-the-art hashing-based techniques differ significantly from Stockmeyer's algorithmic procedure and, in particular, exploit the availability of SAT solvers. Yet, they make $O(\log n)$ calls to a SAT oracle, which coincides with the number of NP oracle calls in Stockmeyer's algorithmic procedure. 
	
		\item[O2]  Stockmeyer provided a matching lower bound of $\Omega(\log n)$ on the number of NP calls, which follows from the simple observation that for a fixed $\varepsilon$, there are $\Theta(n)$ possible outputs that an algorithm can return. Since every NP call returns an answer, Yes or No, the trace of an algorithm can be viewed as a binary tree such that every leaf represents a possible output value. Therefore, the height of the tree (i.e., the number of NP calls) must be $\Omega(\log n)$. Since a SAT oracle returns a satisfying assignment (i.e., provides $n$ bits of information), the trace of the algorithm is no longer a binary tree, and therefore, Stockmeyer's analysis does not extend to the case of SAT oracles for approximate model counting. 
\end{description}

To summarize, the best-known upper bound for SAT oracle calls for approximate model counting is $O(\log n)$, which matches the upper bound for NP oracle calls. However, the technique developed in the context of achieving a lower bound for NP oracle calls does not apply to the case of SAT oracle. Therefore, one wonders whether there exist algorithms with a lower number of SAT oracle calls. In other words, are SAT oracles more powerful than NP oracles for the problem of approximate model counting? 

The primary contribution of this work is to resolve the above challenge. In contrast to the problem of uniform sampling, we reach a starkly different conclusion: SAT oracles are no more powerful than NP oracles in the context of approximate model counting. Formally, we prove the following theorem:

\begin{restatable}{thm}{mainthm}
\label{thm:main} For any $\epsilon, \delta\in (0,1)$, 
given a formula $\formula$, computation of $(\varepsilon,\delta)$-approximation of $|\satisfying{\formula}|$ requires $\tilde{\Omega}(\log n)$\footnote{The tilde hides a factor of $\log\log n$} queries to a SAT oracle. 
\end{restatable}


The establishment of the above theorem turned out to be highly challenging as the existing approaches in the context of NP oracles are not applicable to the SAT oracles. We provide an overview of our approach below.  

\subsection{Technical Overview}

In order to provide the lower bound on the number of queries required by the SAT oracle, we work with a stronger SAT oracle model. In particular, an answer from a (standard) SAT oracle does not provide any extra guarantee/information other than that the returned assignment is a satisfying assignment of the queried formula. Our lower bound works even if we consider that the returned satisfying assignment is chosen randomly from the set of satisfying assignments. More specifically, we consider a stronger model, namely \emph{$\sat$ oracle}, which returns a uniformly chosen solution of a queried formula $\formula$ whenever the formula is satisfiable. It is worth remarking that while a SAT oracle can be simulated by only $n$ queries to an NP oracle, the best-known technique to simulate $\sat$ makes $O(n^2 \log n)$ queries to an NP oracle~\cite{BGP00, delannoy2022almost}. We prove the following theorem which implies Theorem~\ref{thm:main}. 

\begin{restatable}{thm}{mainthm2}
\label{thm:main2} For any $\epsilon< 1/2$ and any $\delta\leq 1/6$, 
given a formula $\formula$, computation of $(\varepsilon,\delta)$-approximation of $|\satisfying{\formula}|$ requires $\tilde{\Omega}(\log n)$ queries to a $\sat$ oracle. 
\end{restatable}
Although we consider $\epsilon< 1/2$ and $\delta\leq 1/6$ in the above theorem and provide the proof accordingly, our proof works even for any constant $\epsilon, \delta \in (0,1)$. Another thing to remark is that in our proof, we allow even exponential (in the size of the original formula) size formula to be queried in the $\sat$ oracle, making our result stronger than what is claimed in the above theorem.

Let us assume that $\algo$ is an algorithm that $(\epsilon, \delta)$-approximates $|\satisfying{\formula}|$ for any given input $\formula$ (on $n$ variables) by making $q$ queries to a $\sat$ oracle. We will refer to such an algorithm as a $\counter$. We would like to prove a lower bound on $q$. 

The main technical difficulty in proving our lower bound results comes from the enormous power of a $\sat$ oracle compared to an NP oracle. An NP oracle can only provide a YES or NO answer, restricting the number of possible answers (from the NP oracle)  to $2^q$ for a  $q$-query counter with an NP oracle. On the other hand,  since a $\sat$ oracle returns a (random) satisfying assignment (if a satisfying assignment exists),  the number of possible answers can be $2^{nq}$. Further, any counter can be adaptive -- it can choose the next query adaptively based on the previous queries made and their corresponding answers. In general, proving a non-trivial (tight) lower bound for any adaptive algorithm turns out to be one of the notorious challenges, and the difficulty in proving such a lower bound arises in other domains like data structure lower bound, property testing, etc. One of the natural ways to prove any lower bound is to use the information-theoretic technique. However, one of the main challenges in applying such techniques in the adaptive setting is that conditional mutual information terms often involve complicated conditional distributions that are difficult to analyze.



To start with, we argue that we can assume that the $\counter$ is "semi-oblivious" in nature. The number of satisfying assignments of a formula does not change by any permutation of the elements in $\tf^n$, and the $\counter$ can only get elements of $\satisfying{\formula}$ by querying the $\sat$ oracle. So we argue that the only useful information of the $i^{th}$ query set (that is, the set of satisfying assignments of the formula that is given to the $\sat$ oracle) is the size of its intersection with the previous $(i-1)$ query sets and their corresponding answers. We formalize it in Section~\ref{sec:core}. 

We next use Yao's minimax principle to prove a lower bound on the number of queries to a $\sat$ oracle made by a deterministic "Semi-oblivious counter" when the input formula $\formula$ is drawn from a "hard" distribution.


For the hard distribution, 
we construct $O(n^{3/4})$ formulas $\formula_{\ell}$ for each value of  $\ell$ in the set $\{\lfloor n^{1/4}\rfloor,\lfloor n^{1/4}\rfloor+1,\dots,\lceil n^{3/4} \rceil \}$. The formulas  $\formula_{\ell}$ are chosen in such a way that $|\satisfying{\formula_{\ell}}| \approx 2 |\satisfying{\formula_{\ell + 1}}|$ thereby approximately counting the number of satisfying assignments (upto a multiplicative $(1+\epsilon)$-factor for small constant $\epsilon$) reduces to the problem of determining the value of $\ell$. The hard distribution is obtained by picking an  $\ell$ uniformly at random from the set  
$\{\lfloor n^{1/4}\rfloor,\lfloor n^{1/4}\rfloor+1,\dots,\lceil n^{3/4} \rceil \}$ and using the corresponding formula $\formula_{\ell}$. 

Finally, we show the lower bound using information theory. 
 At a high level, we show that the information gained about $\ell$ by the knowledge of obtained samples is small unless we make $\Tilde{\Omega}(\log n)$ oracle calls (Lemma~\ref{lem:bound-mutual-inf-overall}). Then we turn to  Fano's Inequality (Theorem~\ref{thm:fano})
  which links the error probability of a counter to the total information gain. 
 Showing that the information gained by samples is small boils down to showing that the KL-divergence of the conditional distribution over the samples is small for all formulas $\formula_{\ell}$ (shown in the proof of the third part of Lemma~\ref{lem:bound-mutual-inf-overall}). 
 The difficulty in showing the above bound comes from the fact that the samples are adaptive and may not always be concentrated around the expectation. To overcome the above challenge, we first define an indicator random variable $Y_i$ to denote whether, at the $i^{th}$ query, the concentration holds (see the definition in Equation~\ref{eq:defY}). Then we split it into cases: In the first case,  we argue for the situation when concentration may not hold at some step of the algorithm (if $Y_i =1$ for some $i \in [q]$). The second case is when concentration holds (if $Y_i =0$ for all $i \in [q]$). We believe that the technique developed in this paper can be a general tool to show sampling lower bounds in a number of other settings.

\section{Notations and Preliminaries}

For any integer $m$, let $[m]$ denote the set of integers $\{1,2,\ldots,m\}$. For a formula $\formula$ over variable set $\vars(\formula) = \{v_1, \dots, v_n\}$, we denote by $\satisfying{\formula}$ the set of satisfying assignments of $\formula$. If $\formula$ is not satisfiable then $\satisfying{\formula} = \emptyset$. We can interpret $\satisfying{\formula}$ as a subset of $\tf^n$. On the other hand, for any subset $A\subset \tf^n$ we denote by $\queryformula_A$ the formula whose set of satisfying assignments is exactly $A$; that is, $\satisfying{\queryformula_A} = A$.

\paragraph*{Oracles and query model}
In our context of Boolean formulas, an \emph{NP oracle} takes in a Boolean formula $\formula$ as input and returns Yes if $\formula$ is satisfiable (i.e., $\satisfying{\formula} \ne \emptyset$), and No, otherwise. Modern SAT solvers, besides determining whether a given formula is satisfiable or not, also return a satisfying assignment (arbitrarily) if the formula is satisfiable. This naturally motivates us to consider an oracle, namely \emph{$\sat$ oracle}, that takes in a Boolean formula $\formula$ as input and, if $\formula$ is satisfiable, returns a satisfying assignment uniformly at random from the set $\satisfying{\formula}$, and $\bot$, otherwise. 

We rely on the query model introduced by Stockmeyer~\cite{stockmeyer1983complexity}: For a given $\formula$ whose model count we are interested in estimating, one can query the corresponding (NP/SAT) oracle with formulas of the form $\formula \wedge \psi_A$, where, as stated earlier, $\psi_A$ is an (arbitrary) formula whose set of solutions is $A$. We will use $\formula_A $ as a shorthand to represent $\formula \wedge \psi_A$. Throughout this paper, we consider the above query model with query access to the $\sat$ oracle. One call to the $\sat$ oracle will be called a $\sat$ query. By abuse of notation, we sometimes say ``$A$ is queried" to refer to the formula $\formula_A$.




\paragraph*{$k$-wise independent hash functions}
Let $n,m,k$ be positive integers and let $H(n,m,k)$ denote the family of \emph{$k$-wise independent hash functions} from $\tf^n$ to $\tf^m$. For any $\alpha \in \tf^m$, and $h \in H(n,m,k)$, let $h^{-1}(\alpha)$ denote the set $\{s \in \tf^n \mid h(s) = \alpha\}$.

It is well-known (e.g., see~\cite{CLRbook}) that for any integer $n,m,k$, one can generate an explicit family of $k$-wise independent hash functions in time and space $poly(n,m,k)$. Moreover, for any $\alpha \in \tf^m$, $h^{-1}(\alpha)$ (where $h \in H(n,m,k)$) can be specified by a Boolean formula of size $poly(n,m,k)$.


\paragraph*{Concentration inequalities for limited independence}

\begin{lemma}~\cite{schmidt1995chernoff}
\label{lem:k-wise-conc-ine}
    If $X$ is a sum of $k$-wise independent random variables, each of which is confined to $[0,1]$ with $\mu =\E[X]$ then
    \begin{enumerate}
        \item For any $\gamma \le 1$ and $k \ge \gamma^2 \mu$, $\Pr[|X - \mu| \ge \gamma \mu] \le exp(-\gamma^2 \mu/3)$.
        \item For any $\gamma \ge 1$ and $k \ge \gamma \mu$, $\Pr[|X - \mu| \ge \gamma \mu] \le exp(-\gamma \mu/3)$.
    \end{enumerate}
\end{lemma}

\paragraph*{Basics of information theory}
Let $X$ and $Y$ be two random variables over the space $\mathcal{X} \times \mathcal{Y}$. 
The mutual information $I(X;Y)$ between random variables $X$ and $Y$ is the reduction in the entropy of $X$ given $Y$ and hence
\begin{equation}\label{eq:mutual-inf-ub-entropy}
I(X;Y) = H(X) - H(X|Y) \le H(X)
\end{equation}
where $H(X) = -\sum_{x \in \mathcal{X}} \Pr[X=x] \log \Pr[X=x]$ is the Shannon entropy of $X$ and $H(X|Y)$ is the conditional entropy of $X$ given $Y$.

The \emph{Kullback–Leibler divergence} or simply \emph{KL divergence} (also called {relative entropy}) between two  discrete probability distributions $P$ and $Q$ defined on same probability space $\mathcal{X}$ is given by :
\[
KL(P||Q) := \sum_{x \in \mathcal{X}}p(x) \log \frac{p(x)}{q(x)}
\]
where $p$ and $q$ are probability mass  functions of $P$ and $Q$ respectively. 

If the joint distribution of $X$ and $Y$ is $Q_{X,Y}$ and marginal distributions $Q_X$ and $Q_Y$ respectively, then the mutual information $I(X;Y)$ can also be equivalently defined as:
\[
I(X;Y) := KL(Q_{X,Y}||Q_X \times Q_Y).
\]

For three random variables $X,Y,Z$, the \emph{conditional mutual information} $I(X;Y|Z)$ is defined as 
\begin{align*}
I(X;Y|Z) :=  \mathbb{E}_Z[KL(Q_{(X,Y)|Z}||Q_{X|Z} \times  Q_{Y|Z})]. 
\end{align*}

For any three random variables $X,Y,Z$, the \emph{chain rule} for mutual information says that
\[
I(X;(Y,Z)) = I(X;Y)+I(X;Z|Y).
\]

If $Z$ is a discrete random variable taking values in $\mathcal{Z}$ then we have 
\begin{align*}
&\mathbb{E}_Z[KL(Q_{(X,Y)|Z}||Q_{X|Z} \times  Q_{Y|Z})] = \sum_{z \in \mathcal{Z}} Q_Z(z) \cdot  KL(Q_{(X,Y)|Z=z}||Q_{X|Z=z} \times Q_{Y|Z=z})\\
& = \sum_{z \in \mathcal{Z}} Q_Z(z) \cdot  I(X;Y|Z=z).
\end{align*}

\begin{lemma}[\cite{scarlett2019introductory}]
\label{lem:kl}
Let $P_X,P_Z,P_{Z|X}$ be the marginal distributions corresponding to a pair $(X,Z)$, where $X$ is discrete. For any auxiliary distribution $Q_Z$, we have 
\[
I(X,Z) = \sum_{x} P_X(x) KL(P_{Z|X}(\cdot|x)||P_Z) \le \max_{x} KL(P_{Z|X}(\cdot|x) || Q_Z).
\]
\end{lemma}

 \begin{theorem}[Fano's inequality]
 \label{thm:fano}
  Consider discrete random variables $X$ and $\hat{X}$, both taking values in $\mathcal{V}$. Then 
  \[
  \Pr[\hat{X} \neq X] \ge 1 - \frac{I(X;\hat{X})+\log 2}{\log |\mathcal{V}|}.
  \]
 \end{theorem}
 
 Consider the random variables $X, Z,\hat{X}$. 
 If the random variable $\hat{X}$ depends only on $Z$ and is conditionally independent on $X$, then we have \begin{equation}\label{eq:dataprocessing}
     I(X;\hat{X}) \le I(X;Z).
 \end{equation} This inequality is known as the \emph{data processing inequality}. For further exposition, readers may refer to any standard textbook on information theory (e.g.,~\cite{CT06}).

\paragraph*{Minimax theorem}
Yao's minimax principle~\cite{yao1977probabilistic} is a standard tool to show lower bounds on the worst-case performance of randomized algorithms. 
Roughly speaking, it says that to show a lower bound on the performance of a randomized algorithm $R$, it is sufficient to show a lower bound on any deterministic algorithm when the instance is randomly drawn from some distribution.

Consider a problem over a set of inputs $\mathcal{X}$. Let $\Gamma$ be some probability distribution over $\mathcal{X}$ and let $X\in \mathcal{X}$ be an input chosen as per $\Gamma$.  
Any randomized algorithm $R$ is essentially a probability distribution over the set of deterministic algorithms, say $\mathcal{T}$. 
By Yao's minimax principle,
\[
\max_{X \in \mathcal{X}} \Pr[R \thinspace \text{gives wrong answer on} \thinspace  X]   \ge  \min_{T \in \mathcal{T}} \Pr_{X \sim \Gamma}[ T \thinspace \text{gives wrong answer on} \thinspace  X].
\]


\section{Lower Bound on the number of queries to $\sat$ oracle}\label{sec:lbproof}

In this section, we will prove Theorem~\ref{thm:main2}, which implies Theorem~\ref{thm:main}.  Let $\algo$ be an adaptive randomized algorithm that given as input $\formula$ over $n$ variables $\vars = \{v_1, \dots, v_n\}$ and output $\est$ that is an $(\epsilon, \delta)$-approximation of $\satisfying{\formula}$. The only way $\algo$ accesses the input $\formula$ is by making queries to the $\sat$ oracle, that is, obtaining random satisfying assignments from $\satisfying{\formula_A}$, where $\formula_A = \formula\wedge \queryformula_A$.  We will prove that $\algo$ has to make at least $\tilde{\Omega}(\log n)$ such queries to the $\sat$ oracle. 

We will start by arguing that we can assume that the adaptive algorithm $\algo$ has some more structure. In particular, in Section~\ref{sec:core}, we will argue (in the same lines as in \cite{CFGM16}) that we can assume $\algo$ is a \emph{semi-oblivious counter} (Definition~\ref{def:core}). 

We use Yao's Min-max technique to argue that obtaining a lower bound on a (randomized) semi-oblivious counter is the same as obtaining a lower bound on a (deterministic) semi-oblivious counter when the input is drawn from the worst possible distribution over the set of formulas on $n$ variables. In Section~\ref{sec:hardinstance}, we present the ``hard" distribution that would help us prove the lower bound against any deterministic semi-oblivious counter. In Section~\ref{sec:hardproperties}, we present some properties of the hard instance that would be used for the final lower-bound proof. 

Finally, in Section~\ref{sec:mainproof}, we will use an information-theoretic argument to give a lower bound on the query complexity of any deterministic semi-oblivious counter and hence prove Theorem~\ref{thm:main2}.

\paragraph*{A note on the use of auxiliary variables in the queries to the  $\sat$ oracle} One thing we observe is that our lower bound proof does not assume that in the input formula $\formula$, all the variables are influential. In other words, we can assume that $\formula$ is on $n$ variables; the actual number of variables in $\formula$ may be significantly less. All we need for our lower bound proofs to go through is that the queries to the $\sat$ oracle made by the algorithm are to $\formula \wedge \queryformula_A$ where the $\queryformula$ is a formula over $n$ variables. And the lower bound on the query complexity that we prove (Theorem~\ref{thm:main2}) is $\tilde{O}(\log n)$. Hence, as long as the number of variables used in the queries to the $\sat$ oracle is at most polynomial in the actual number of variables in the input formula $\formula$, our lower bound holds.

\subsection{Semi-oblivious counter}\label{sec:core}

Suppose given a formula $\formula$ over $n$ variables, a counter $\algo$ makes $q$ calls to the $\sat$ oracle with queried formulas $\formula_{A_1},\cdots,\formula_{A_q}$ respectively, where each $A_i \subseteq \tf^n$. (Recall, $\formula_{A_i}=\formula \wedge \psi_{A_i}$, where $\psi_{A_i}$ denote the formula having $\satisfying{\psi_{A_i}}=A_i$.) Note, the $i$-th $\sat$ oracle call by the counter $\algo$ is specified by the set $A_i$. During the $i$-th call (for $1\le i \le q$), suppose the counter $\algo$ receives a sample $s_i \in A_i \cup \{\bot\}$. Note that the oracle calls made by $\algo$ can be \textit{adaptive}, i.e., the sets  $A_1,\cdots, A_q$ are not fixed in advance -- the counter $\algo$ fixes $A_i$ only after seeing the samples $s_1,\cdots,s_{i-1}$ (outcomes of all the previous oracle calls). 

We now define a special type of randomized $\counter$, referred to as \emph{semi-oblivious counter}, which at any point of time queries  the $\sat$ oracle only by looking into the configuration of the previous step. We will later argue that to prove a query lower bound for general $\counters$, it suffices to consider semi-oblivious counters. In other words, semi-oblivious counters are as "powerful" as general $\counters$. 

We first provide intuition for \emph{semi-oblivious counter}.
Note that permuting the variables of any formula $\formula$ permutes the set of satisfying assignments $\satisfying{\formula}$ but $|\satisfying{\formula}|$ is unchanged. Since a $\counter$ needs to determine $|\satisfying{\formula}|$ only (not $\satisfying{\formula}$),  the final output by the $\counter$, in some sense, should be based only on the relations between the samples and the query sets (not on their actual values). 
Before providing a formal definition, let us first introduce some terminology.

Given a family of sets $\mathcal{A} = \{A_1,\cdots,A_i\},$ (where $A_i \subseteq \tf^n$), the \emph{atoms} generated by $\mathcal{A}$, denoted by $\At(\mathcal{A})$, are (at most) $2^i$ distinct sets of the form $\cap_{j=1}^{i}C_j$ where $C_j \in \{A_j,\tf^n\setminus A_j\}$.  For example, if $i =2$, then $\At(A_1, A_2) = \{A_1\cap A_2, A_1\setminus A_2, A_2\setminus A_1, (A_1\cup A_2)^c\}$.





\begin{definition} \label{def:core}
    (Semi-oblivious counter): A semi-oblivious counter is a randomized algorithm $T$ that, given any formula $\formula$, at any step $i$, works in the following three phases: 
    \begin{itemize}
        \item \textbf{Semi-oblivious choice:} Let $\mathcal{A}_{i-1}=\{A_1,\cdots,A_{i-1}\}$, $S_{i-1}=\{s_1,\cdots,s_{i-1}\}$, $C_{i-1}=\{c_1,\cdots,c_{i-1}\}$ be the set of first $i-1$ query sets, the set of first $i-1$ samples obtained, the set of first $i-1$ configurations, respectively. 
        Only based on $C_{i-1}$ (without knowing the set $S_{i-1}$), $T$ does the following:
        \begin{itemize}
            \item For each $A \in \At(\mathcal{A}_{i-1})$, it generates an integer $k^A_i$ between $0$ and $|A \setminus S_{i-1}|$. ($k^A_i$ indicates how many unseen elements from the atom $A$ of the previous query sets are to be included in the next query set.)
            \item It chooses a set of indices $K_i \subseteq \{1,\cdots,i-1\}$. ($K_i$ specifies the index set of previous samples that are to be included in the next query set.)
        \end{itemize}
        \item \textbf{Query set generation:} In this phase, it decides the query set $A_i$ as follows:
        \begin{itemize}
            \item Let us define the \emph{candidate unseen set family} as 
            \[
            \mathcal{U}_i := \{ U \subseteq \tf^n\setminus S_{i-1} \mid \forall A \in \At(\mathcal{A}_{i-1}), |U_i \cap A| = k^A_i \}.
            \]
            The algorithm $T$ chooses a set $U_i$ uniformly at random from the candidate unseen set family $\mathcal{U}_{i-1}$.

            \item Let us denote $O_i:= \{s_j \mid j \in K_i\}$. The algorithm $T$ decides the query set to be $A_i = U_i \cup O_i$.
        \end{itemize}
        \item \textbf{Oracle call:} It places a query to the $\sat$ oracle with the formula $\formula_{A_i}$. Let the $i$-th \emph{configuration} $c_i$ specify whether $s_i = \bot$, or for which $j \in K_i$, $s_i=s_j$, or for which $A \in \At(\mathcal{A}_{i-1})$, $s_i \in A\cap U_i$.
    \end{itemize}
    In the end (after placing $q=q(n)$ $\sat$ oracle calls), depending on the set of all the configurations $C_q$, $T$ outputs an estimate on the $|\satisfying{\formula}|$.
\end{definition}

From now on, for brevity, we use $\At(U_i)$ to denote the set $\{U_i \cap A \mid A \in \At(\mathcal{A}_{i-1})\}$. Next, we show that if there exists a general $\counter$, then there also exists a semi-oblivious counter. The proof is inspired by the argument used in~\cite{CFGM16} and is given in Appendix~\ref{sec:oblivious-equivalence}. 

\begin{lemma}\label{obs:core}
If there is an algorithm that, given any input $\formula$ on $n$ variables, outputs an $(\epsilon, \delta)$-approximation of $|\satisfying{\formula}|$ while placing at most $q=q(n)$ $\sat$ oracle calls, then there also exists a (randomized) semi-oblivious counter that, given input $\formula$, outputs an $(\epsilon, \delta)$-approximation of $|\satisfying{\formula}|$ while also placing at most $q$ $\sat$ oracle calls.
\end{lemma}

Suppose all the internal randomness of a semi-oblivious counter is fixed. (Since in the proof of Theorem~\ref{thm:main}, we will first apply Yao's minimax principle, it suffices  to only consider deterministic decision trees.) Then, a semi-oblivious counter $T$ can be fully described by a decision tree $R$ where the path from the root to any node $v$ at depth $i$ (more precisely, the edges of this path) corresponds to the configuration of the first $i-1$ samples. Note that fixing the configurations of the samples till $i-1$ queries (and the internal randomness) fixes the size of an atom $A \in At(A_1,\cdots,A_i)$ (and hence of each $A_j$ for $j \le i$). 
Formally,
\begin{enumerate}
    \item[(i)] A path (from root) to any node $v$ at depth $i$ is associated with a sequence of query sets $ \mathcal{A}_{i-1}=(A_1,\cdots,A_{i-1})$ such that the sizes of all atoms $A \in  \At(\mathcal{A}_{i-1})$ are fixed.
    \item[(ii)] The node $v$ is labeled by a vector $\bm{k}_v = (k^A_i)_{A \in \At(\mathcal{A}_{i-1})}$ and a set $K_v \subseteq [i-1]$ which are used to determine the next query set $A_i = O_i \cup U_i$. (Again, $|U_i| = \sum_{A \in \At(\mathcal{A}_{i-1})} k^A_i$ and the set $U_i$ is fixed.) $A_i$ is used to place the next $\sat$ oracle call.
    \item[(iii)] For every possible value of the configuration at step $i$, there is a corresponding child of the node $v$, with the corresponding edge labeled by the value of the configuration. 
\end{enumerate}

For any node $v$, we use $A_v = O_v \cup U_v$ to denote the (random) query set (corresponding to the node $v$) determined by the  $\bm{k}_v$ and $K_v$. Note that $|U_v| = \sum_{A \in \At(\mathcal{A}_{i-1})} k^A_i$. Further, we use $\mathcal{A}_{v} := (A_1,\cdots,A_v)$ for the sequence of query sets corresponding to a path to $v$ and node $v$. Observe the number of possible outcomes of the counter $T$ at any step $i$ is at most $i + 2^i +1 \le 2^{q+1}$ (since $i \le q$). So the total number of nodes in the decision tree corresponding to the semi-oblivious counter $T$ is at most $2^{O(q^2)}$.

\subsection{Hard instance} \label{sec:hardinstance}
We will provide a set of inputs $\mathcal{X}$ (which, in our case, will be a set of formulas) and a distribution $\Gamma$ over $\mathcal{X}$. Then we will show that any deterministic semi-oblivious counter $D$ (note that $D$ knows $\mathcal{X}$ and $\Gamma$) which receives as input a formula $\formula\in \mathcal{X}$ randomly drawn as per distribution $\Gamma$ and  returns an $(\epsilon,\delta)$-approximation of $\satisfying{\formula}$, must make $\tilde{\Omega}(\log n)$  queries to the $SAT$-oracle.

Let $k = (\log n)^9$. Let $\mathcal{X}$ be the set of all formulas (with $n$ variables). We now define the hard distribution $\Gamma$ over $\mathcal{X}$ as follows by describing the procedure of picking a formula in $\mathcal{X}$ according to $\Gamma$. 
\begin{enumerate}
\item Pick $\ell \in \{\lfloor n^{1/4}\rfloor,\lfloor n^{1/4}\rfloor +1,\dots,\lceil n^{3/4}\rceil\}$ uniformly at random. 
\item Draw  a hash function $h_{\ell} \leftarrow H(n,\ell,k)$ uniformly at random.  
\item  Let $\formula_{\ell}$ be the formula whose set of satisfying assignments is $h_{\ell}^{-1}(F^\ell)$. (Recall, $h_\ell: \tf^n \to \tf^\ell$.)
\item The formula $\formula_{\ell}$ 
is the picked formula.
 \end{enumerate}



\subsubsection{Properties of the hard instance}\label{sec:hardproperties}

Let  $f_{\ell} :=  \E[|\satisfying{\formula_{\ell}}|] = \E[|h_{\ell}^{-1}(F^{\ell})|]$ for  $\ell \in \{\lfloor n^{1/4}\rfloor,\lfloor n^{1/4}\rfloor+1,\dots,\lceil n^{3/4}\rceil\}$. Observe, it follows from the construction of $\formula_{\ell}$ and the properties of hash functions that $f_{\ell} = \frac{2^n}{2^{\ell}}$.

\begin{lemma}
\label{lem:supportofi}
    With  probability at least $1 - n 2^{-n/20}$, 
    we have \begin{equation}\label{eq:supportofi}
       \mbox{ for all } \ell,\ ||\satisfying{\formula_{\ell}]}| - f_{\ell}| \le 2^{-n/10} f_{\ell}.
    \end{equation}
\end{lemma}
\begin{proof}
  It is straightforward to see that the variance of $|\satisfying{\formula_{\ell}}|$ is $Var[|\satisfying{\formula_{\ell}}|] \le f_{\ell}$. So by Chebyshev's inequality, 
  \[
  \Pr\left[||\satisfying{\formula_{\ell}}| - f_{\ell}| \ge 2^{-n/5} f_{\ell}\right] \le \frac{2^{n/5}}{f_{\ell}} \le \frac{2^{n/5}\cdot 2^{\ell}}{2^n}\le 2^{-n/20}.
  \]
  The lemma now follows from a union bound over all $\ell$.
\end{proof}




\begin{definition} Once $\ell\in \{\lfloor n^{1/4}\rfloor, \lfloor n^{1/4}\rfloor +1, \dots, \lceil n^{3/4}\rceil\}$ has been picked in Step 1 of the construction of the hard instance (Section~\ref{sec:hardinstance}), let for any $S \subseteq \tf^n$
$$\num{S}{\ell} = \mathbb{E}\left[|\satisfying{\formula_{\ell}} \cap S|\right],$$ 
where the expectation is over the choice of the hash function is Step 2 of the construction of the hard instance.
 \end{definition}

Note that for any $S\subseteq \tf^n$ the value of $\num{S}{\ell}$ is $|S|/2^{\ell}$.
 

\begin{lemma}
\label{lem:conc-ine}
With   probability at least $1 - \frac{2^{O(q^2)}}{n^{(\log n)^4}}$, 
the following holds: 

For any node $v$ in the decision tree $R$ and any atom $A \in \At(U_v)$,
\begin{enumerate}
    \item If $\num{U_v}{\ell}  < \frac{1}{n^{(\log n)^4}}$ then $|U_v \cap \satisfying{\formula_{\ell}}| = 0$. Similarly, if $\num{A}{\ell}  < \frac{1}{n^{(\log n)^4}}$ for any atom $A \in \At(U_v)$ then $|A\cap \satisfying{\formula_{\ell}}| = 0$
    \item  If  $\num{U_v}{\ell} \ge (\log n)^5$ then $\frac{1}{2} \num{U_v}{\ell} \leq |U_v \cap \satisfying{\formula_{\ell}}| \leq \frac{3}{2} \num{U_v}{\ell}$. Similarly, if  $\num{A}{\ell} \ge (\log n)^5$ then $\frac{1}{2}\num{A}{\ell} \leq  |A \cap \satisfying{\formula_{\ell}}| \leq \frac{3}{2} \num{A}{\ell}$
    \item If $\num{U_v}{\ell} \le (\log n)^5$ then $|U_v \cap \satisfying{\formula_{\ell}}| \le 2(\log n)^{5}$. Similarly, if $\num{A}{ \ell} \le (\log n)^5$ then $|A \cap \satisfying{\formula_{\ell}}| \le 2(\log n)^{5}$.
\end{enumerate}
\end{lemma}
\begin{proof}
  From Markov's inequality, we have 
  \begin{align*}
      \Pr[|U_v \cap \satisfying{\formula_{\ell}}| \ge 1] \le \Pr\left[|U_v \cap \satisfying{\formula_{\ell}}| \ge \left(\frac{1}{\num{U_v}{\ell}}-1\right)\num{U_v}{\ell}\right] \le 2 \num{U_v}{\ell}
  \end{align*}
   Taking a union bound over all nodes $v$ with  $\num{U_v}{\ell}  < \frac{1}{n^{(\log n)^4}}$ and all possible values of $\ell$ (which can take $O(n^{3/4})$ values), we get the first part.
   
   From the first part of the Lemma~\ref{lem:k-wise-conc-ine}, by setting $\gamma =1/2$, we have  
    \begin{align*}
      \Pr[|U_v \cap \satisfying{\formula_{\ell}}| \ge  \num{U_v}{\ell}]   \le \exp\left(-\frac{\num{U_v}{\ell}}{12}\right)
  \end{align*}
  for all nodes $v$ in $R$ such that $\num{U_v}{\ell} \ge (\log n)^5$ (note that we have $k = (\log n)^9 > \gamma^2 \num{U_v}{\ell}$). 
  Taking a union bound over all such nodes $v$  and all possible values of $\ell$, we get the second bound.
  
  Let $\gamma_v = \frac{(\log n)^5}{\num{U_v}{\ell}}$. Since $k = (\log n)^9 > \gamma_v \num{U_v}{\ell}$, from the second part of 
 Lemma~\ref{lem:k-wise-conc-ine}, we have 
   \begin{align*}
      \Pr[|U_v \cap \satisfying{\formula_{\ell}}| \ge \gamma_v \num{U_v}{\ell}]   \le \exp\left(-\gamma_v \frac{\num{U_v}{\ell}}{3}\right) \le O\left(\frac{1}{n^{(\log n)^4}}\right).
  \end{align*}
  for all nodes $v$ such that $\num{U_v}{\ell} \le (\log n)^5$.
   Taking a union bound over all such nodes $v$ and all possible values of $\ell$, we get the third part.
\end{proof}

\subsection{Proof of Theorem~\ref{thm:main2}}\label{sec:mainproof}



\begin{proof}[Proof of Theorem~\ref{thm:main2}]
By Lemma~\ref{obs:core} and Yao's minimax theorem, we can assume that our $\counter$ $\algo$ is a (deterministic) semi-oblivious counter whose input is a randomly chosen formula $\formula \in \mathcal{\formula}_n$, as per distribution $\Gamma$ and $\algo$ returns $\est$ which is an $(\epsilon, 2/3)$-approximation of $|\satisfying{\formula}|$. 
We will prove that $\algo$ must make $q =\Tilde{\Omega}(\log n)$ many $SAT$-oracle calls. 

Recall the distribution $\Gamma$ (Section~\ref{sec:hardinstance}) over the set of all formulas. We can assume that the input to $\algo$ is $\formula_{\ell}$, where  $\ell$ is uniformly drawn from the set $\{\lfloor n^{1/4}\rfloor,\lfloor n^{1/4}\rfloor+1,\dots,\lceil n^{3/4}\rceil\}$. 

Consider the path taken by the semi-oblivious counter $\algo$ in the decision tree. 
Let the $i$th query made by $\algo$ (that is at vertex $v_i$) be $A_i = U_i \cup O_i$ (as in Definition~\ref{def:core}).
Let $Z_i$ be the configuration (denoted as $c_i$ in Definition~\ref{def:core}) of the sample from $A_i$. Note that the domain of $Z_i$ is $\Omega_i := O_i \cup \At(U_i) \cup \bot$.

Let $\good$ be the event that the condition in Equation~\ref{eq:supportofi} (in Lemma~\ref{lem:supportofi}) and the condition in Lemma~\ref{lem:conc-ine} holds. 
Note that by Lemma~\ref{lem:supportofi} and Lemma~\ref{lem:conc-ine}  if $q \leq \log n$ then 
\begin{equation}\label{eq:good}
    \Pr[\good] = 1-o(1).
\end{equation}

Let $X$ be the random variable that takes values in $\{\lfloor n^{1/4}\rfloor, \lfloor n^{1/4}\rfloor +1 , \dots, \lceil n^{3/4}\rceil\}$ uniformly at random (in Step 1 of the construction of hard instance). 
 Note that by the triangle inequality
 \begin{equation}\label{eq:triangle}
 \left|\est - |\satisfying{\formula_{\ell}}|\right| \geq \left|\est - \frac{2^n}{2^{\ell}}\right| - \left|\frac{2^n}{2^{\ell}} - |\satisfying{\formula_{\ell}}|\right|.\end{equation}
 By Lemma~\ref{lem:supportofi} we know that with probability at least $(1-1/6)$,  we have $|\frac{2^n}{2^{\ell}} - |\satisfying{\formula_{\ell}}|| \leq \frac{1}{2^{n/10}}\cdot \frac{2^n}{2^{\ell}}.$
On the other hand, since $\algo$ outputs an $(\epsilon, \delta)$-approximation of $|\satisfying{\formula}|$ (with $\epsilon < 1/2$ and $\delta < 1/6$),  Equation~\ref{eq:triangle} implies that with probability at least $(1-\frac{1}{6} - \delta) \geq \frac{2}{3}$ we have 
\begin{equation}\label{eq:ell}
    \left|\est - \frac{2^n}{2^{\ell}}\right| \leq \left(\epsilon + \frac{1}{2^{n/10}}\right)\frac{2^n}{2^{\ell}} \leq \frac{1}{2}\cdot \frac{2^n}{2^{\ell}},
\end{equation}
where the last inequality follows from the fact that $\epsilon \leq 1/3$. Since $|\frac{2^n}{2^{\ell}} - \frac{2^n}{2^{\ell'}}| > \frac{1}{2}\cdot \frac{2^n}{2^{\ell}}$ for any integer $\ell' \neq \ell$, so  Equation~\ref{eq:ell} is satisfied only when $\hat{X}$ is same as the picked $\ell$ (that is $\hat{X} = X$) where,
$$\hat{X} =  \argmin_{\ell \in \{\lfloor n^{1/4}\rfloor,\lfloor n^{1/4}\rfloor+1,\dots, \lceil n^{3/4} \rceil\} }  \left|\frac{2^n}{2^{\ell}} - \est\right| .$$


Hence, assuming $\good$
\begin{equation}\label{eq:Xhat}
    \frac{1}{3} \geq \Pr[\hat{X} \neq X].
\end{equation}

By Fano's Inequality (Theorem~\ref{thm:fano}) 
\begin{equation}\label{eq:fano2}
    \Pr[\hat{X} \neq X] \geq  1 - \frac{I(X;\hat{X})}{O(\log n)}
\end{equation}

Since the final outcome of the algorithm is determined by the outcome at each step, i.e., $\bm{Z}=(Z_1, \dots, Z_q)$, so by the data processing inequality (Equation~\ref{eq:dataprocessing}), we have 
\begin{equation}\label{eq:data2}
    I(X;\hat{X}) \le I(X;Z_1, \dots, z_q).
\end{equation}

Let $Y_i$ be the random variable that defined as 
\begin{equation}
    \label{eq:defY}
    Y_i =
    \left\{
	\begin{array}{ll}
		1  & \mbox{if } \frac{1}{n^{(\log n)^4}}\leq \num{U_{i}}{\ell}   \leq n^{(\log n)^4} \\ 
		0 & \mbox{otherwise } 
	\end{array}
\right.
\end{equation}

Again by the data-processing inequality (Equation~\ref{eq:dataprocessing}), we have \begin{equation}\label{eq:data3}
    I(X;Z_1,\dots,Z_q) \le I(X;Y_1,Z_1,\dots,Y_q,Z_q).
\end{equation}

By the chain rule of mutual information, we have
\begin{equation}\label{eq:mutual}
I(X;Y_1,Z_1,\dots,Y_q,Z_q)
= \sum_{i \in [q]}I(X;Y_i,Z_i|Y_1,Z_1,\dots,Y_{i-1},Z_{i-1}) 
\end{equation}

Finally, we will show, in the following lemma, that conditioned on the fact $\good$ happens, we can upper bound $I(X;Y_1,Z_1,\dots,Y_q,Z_q)$ by $O(\log\log n)$.

\begin{lemma}
    \label{lem:bound-mutual-inf-overall}
    $I(X;(Y_1,Z_1,\dots,Y_q,Z_q)) \le q ( O(\log \log n) + O(\log q) + \frac{2^{2q} poly(\log n)}{n^{(\log n)^3}})$.
    
\end{lemma}

We defer the proof of Lemma~\ref{lem:bound-mutual-inf-overall} and complete the proof of Theorem~\ref{thm:main2} assuming Lemma~\ref{lem:bound-mutual-inf-overall}.

From the Equations~\ref{eq:Xhat}, \ref{eq:fano2}, \ref{eq:data2}, \ref{eq:data3} and Lemma~\ref{lem:bound-mutual-inf-overall}, we have that assuming $\good$ happens
\begin{align*}
    \frac{1}{3} \geq & \Pr[\hat{X} \neq X] && \mbox{[From Equation~\ref{eq:Xhat}] } \\
    \geq &  1 - \frac{I(X;\hat{X})}{O(\log n)} && \mbox{[From Equation~\ref{eq:fano2}]}\\
    \geq &  1 - \frac{I(X;Z_1, \dots, z_q)}{O(\log n)} && \mbox{[From Equation~\ref{eq:data2}]}\\
    \geq & 1 - \frac{I(X;Y_1,Z_1,\dots,Y_q,Z_q)}{O(\log n)} && \mbox{[From Equation~\ref{eq:data3}]}\\    
    \geq & 1 - \frac{I(X;Y_1,Z_1,\dots,Y_q,Z_q)}{O(\log n)} && \mbox{[From Equation~\ref{eq:mutual}]}\\
    \ge & 1 - \frac{q \log \log n}{\log n} && \mbox{[From Lemma~\ref{lem:bound-mutual-inf-overall}]}
\end{align*}

 Thus, from Equation~\ref{eq:good}, if $q \leq \log n$ 
 $$1 - \frac{q \log \log n}{\log n} \leq \frac{1}{3} + \Pr[\good] \leq \frac{1}{3} + O(1)$$
which implies $$q = \Omega\left(\frac{\log n}{\log \log n}\right).$$

\end{proof}

\subsubsection{Proof of Lemma~\ref{lem:bound-mutual-inf-overall}}

\begin{lemma}
\label{lem:bound-mutual-inf} The following holds:
 \begin{enumerate}
     \item  Conditioned on event that $Y_j = 1$ for some $j \le i$, $$I(X;Z_i|Y_1,Z_1,\dots,Y_{i-1},Z_{i-1},Y_i) \le   O(\log \log n),$$
     \item $I(X,Y_i|Y_1,Z_1,\dots,Y_{i-1},Z_{i-1}) \le 1$,
     \item Conditioned on the event that $Y_1= 0, \dots, Y_{i-1} =0$, $$I(X,Z_i|Y_1,Z_1,\dots,Y_{i-1},Z_{i-1},Y_i) \le  O(\log q) + \frac{2^{2q} poly(\log n)}{n^{(\log n)^3}}.$$
 \end{enumerate}
\end{lemma}
\begin{proof} We will prove Part 1, 2, and 3 one by one. 

\ 

 \noindent \textbf{Proof of Part 1: } We will prove that conditioned on event that $Y_j = 1$ for some $j \le i$, $$I(X;Z_i|Y_1,Z_1,\dots,Y_{i-1},Z_{i-1},Y_i) \le   O(\log \log n).$$
From (\ref{eq:mutual-inf-ub-entropy}),   we have $$I(X,Z_i|Y_1,Z_1,\dots,Y_{i-1},Z_{i-1},Y_i) \le H(X|Y_1,Z_1,\dots,Y_{i-1},Z_{i-1},Y_i).$$ Note that if $Y_j = 1$ 
then by definition of $Y_j$ we have $\frac{1}{n^{(\log n)^4}} \leq \frac{|U_j|}{2^{\ell}} \leq  n^{(\log n)^4}$, that is, $$\frac{|U_j|}{n^{(\log n)^4}} \leq 2^{\ell} \leq  |U_j| n^{(\log n)^4}.$$ Note that by  definition of the semi-oblivious counter the sets $|U_1|,\dots,|U_i|$ are deterministically determined by $Z_1,\dots,Z_i$ . Thus, there are $O(\log (n^{(\log n)^4})) = O((\log n)^5)$ possible values of $\ell$ and hence $$H(X|Y_1,Z_1,\dots,Y_{i-1},Z_{i-1}) \le O(\log \log n).$$ This proves the first part. 

\ 

\noindent \textbf{Proof of Part 2: } 
Since $Y_i$ can take only binary values,  we have $$I(X,Y_i|Y_1,Z_1,\dots,Y_{i-1},Z_{i-1}) \le 1.$$ 
This proves Part 2.

\ 

\noindent \textbf{Proof of Part 3:}
We will now prove the upper bound on $I(X;(Y_i,Z_i)|Y_1,Z_1,\dots,Y_{i-1},Z_{i-1})$ for each $i \in [q]$, conditioned on $Y_j = 0$ for all $j \in [i]$.

Note that $Z_1,\dots,Z_{i-1}$ fixes the size of $O_i$ and each atoms in $\At(U_i)$. Note that the domain of $Z_i$, i.e., $\Omega_i$ is $\bot \cup O_i \cup \At(U_i)$. Let $r = |O_i|+2 \le q+2$. 

We  define an auxiliary distribution $Q_{(Y_i, Z_i)}$ as follows: 
 $$Q_{(Y_i,Z_i)}(y_i,z_i) := Q_{Y_i}(y_i) Q_{Z_i|Y_i}(z_i|y_i)$$ where,  
 $Q_{Y_i}(0) = Q_{Y_i}(1) = 1/2$ and

\[
Q_{Z_i|Y_i}(z_i|y_i) =
\begin{cases}
\frac{1}{r}, & \quad z_i \in O_i \cup \bot\\
\frac{1}{r} \cdot \frac{|z_i|}{|U_i|}, & \quad  z_i \in \At(U_i)\\
\end{cases}
\]

Let $P_X,P_Z,P_{Z|X}$ be the marginal distributions corresponding to a pair $(X,Z)$.
Conditioned on $Y_j = 0$ for all $j \in [i]$ and  $Z_j = z_j$ for all $j \in [i-1]$ for any $(z_1,\dots,z_{i-1}) \in \Omega_1 \times \dots \times \Omega_{i-1}$, we have for any $\ell \in \mathcal{X}$ (note that, for brevity, we have ignored the  conditioning on $Y_1,Z_1,\dots,Y_{i-1},Z_{i-1}$, in the expression below)
\begin{align}
\label{eq:klupper}
    & KL(P_{Z_i|X}(\cdot|X = \ell) || Q_{Z_i}) = \sum_{z_i \in \Omega_i} P_{Z_i|X}(z_i|X = \ell) \log \frac{P_{Z|X}(z_i|X = \ell)}{Q_{Z_i}(z_i)}
\end{align}

 Note that if $z_i \in \bot \cup O_i $ then  $Q_{Z_i}(z_i) = \frac{1}{r} \ge \frac{1}{q+2}$. Hence, $$\frac{P_{Z_i|X}(z_i|X = \ell)}{Q_{Z_i}(z_i)} \le q+2 \le 2q.$$
 
 Now we  consider the case when $z_i \in \At(U_i)$. 
 
 If $\num{z_i}{\ell} \ge (\log n)^5$ then from Lemma~\ref{lem:conc-ine} we have $$P_{Z_i|X}(z_i|X = \ell) = \frac{|z_i \cap \satisfying{\formula_{\ell}}|}{|U_i \cap \satisfying{\formula_{\ell}}|} \le 3 \num{z_i}{\ell}/\num{U_i}{\ell}.$$ Note that $$Q_{Z_i}(z_i) = \frac{1}{r} \cdot \frac{|z_i|}{|U_i|} \ge 2q \num{z_i}{\ell}/\num{U_i}{\ell}.$$  Therefore, we have $$\frac{P_{Z_i|X}(z_i|X = \ell)}{Q_{Z_i}(z_i)} \le  O(q).$$ For the case when  $\num{z_i}{\ell} < \frac{1}{n^{(\log n)^4}}$,  we have $|z_i \cap \satisfying{\formula_{\ell}}| = 0$.
 Hence the sum $$\sum_{z_i} P_{Z_i|X}(z_i|X = \ell) \log \frac{P_{Z|X}(z_i|X = \ell)}{Q_{Z_i}(z_i)}$$ when, $z_i \in \bot \cup O_i $ or $z_i \in \At(U_i)$ such that $\num{z_i}{\ell} \ge (\log n)^5$ or $\num{z_i}{\ell} < \frac{1}{n^{(\log n)^4}}$,  is at most $O(\log q)$.

Now we bound the sum $$ \sum_{z_i } P_{Z_i|X}(z_i|X = \ell) \log \frac{P_{Z|X}(z_i|X = \ell)}{Q_{Z_i}(z_i)}$$ when $z_i \in \At(U_i)$ such that $$\frac{1}{n^{(\log n)^4}} < \num{z_i}{\ell} < (\log n)^5.$$

  If $\num{z_i}{\ell} \le (\log n)^5$ then we have $$|z_i \cap \satisfying{\formula_{\ell}}| \le 2(\log n)^5$$ and thus $$P_{Z_i|X}(z_i|X = \ell) \le \frac{4(\log n)^5}{\num{U_i}{\ell}}.$$ Note that $$Q_{Z_i}(z_i) = \frac{1}{r} \cdot \frac{|z_i|}{|U_i|} \ge \frac{1}{2q} \num{z_i}{\ell}/\num{U_i}{\ell}.$$  Hence,  $$\frac{P_{Z|X}(z_i|X = \ell)}{Q_{Z_i}(z_i)} \le O(q (\log n)^5/ \num{z_i}{\ell}).$$
  
  Therefore, we have 
  \begin{align*}
      \sum_{z_i:\frac{1}{n^{(\log n)^4}} < \num{z_i}{\ell} \le (\log n)^5} P_{Z_i|X}(z_i|X = \ell) \log \frac{P_{Z|X}(z_i|X = \ell)}{Q_{Z_i}(z_i)} \\
       < \sum_{z_i: \frac{1}{n^{(\log n)^4}} < \num{z_i}{\ell} \le (\log n)^5}  \frac{4(\log n)^5}{\num{U_i}{\ell}} \log (2q (\log n)^5/ \num{z_i}{\ell})\\
       \le 2^q  \frac{8(\log n)^5}{n^{(\log n)^4}} \log (2q (\log n)^5 n^{(\log n)^4})\\
       \le \frac{2^{2q} poly(\log n)}{n^{(\log n)^4}}. 
  \end{align*}
  The second last inequality follows because there are at most $2^q$ possible values of such $z_i$, $\num{U_i}{\ell} \ge n^{(\log n)^3}/2$ and $\num{z_i}{\ell} \ge \frac{1}{n^{(\log n)^3}}$.

  Now by Lemma~\ref{lem:kl} conditioned on the event that $Y_j =0$ for all $j \le i$ we have
  $$I(X;Z_i|Y_1,Z_1,\dots,Y_{i-1},Z_{i-1},Y_i) \le  KL(P_{Z_i|X}(\cdot|X = \ell) || Q_{Z_i}) \le O(\log q) + \frac{2^{2q} poly(\log n)}{n^{(\log n)^3}}.$$
  
 

\end{proof}

\begin{proof}[Proof of Lemma~\ref{lem:bound-mutual-inf-overall}]
We will first prove that for any $i$ $$I(X;(Y_i,Z_i)|Y_1,Z_1,\dots,Y_{i-1},Z_{i-1}) \leq O(\log \log n) + O(\log q) + \frac{2^{2q} poly(\log n)}{n^{(\log n)^3}}.$$
By the chain rule of mutual information,  
    \begin{align*}
        & I(X;(Y_i,Z_i)|Y_1,Z_1,\dots,Y_{i-1},Z_{i-1}) \\ 
        = & I(X;Y_i|Y_1,Z_1,\dots,Y_{i-1},Z_{i-1}) + I(X;Z_i|Y_1,Z_1,\dots,Y_{i-1},Z_{i-1},Y_i) \\ 
        \le & O(\log \log n) + O(\log q) + \frac{2^{2q} poly(\log n)}{n^{(\log n)^3}},
    \end{align*} 
    where the last inequality follows from Lemma~\ref{lem:bound-mutual-inf}.

    Again by the chain rule of mutual information, we have 
    \begin{align*}
    & I(X;(Y_1,Z_1,\dots,Y_q,Z_q)) \\ 
    = & \sum_{i=1}^{q} I(X;(Y_i,Z_i)|Y_1,Z_1,\dots,Y_{i-1},Z_{i-1}) 
    \le  q ( O(\log \log n) + O(\log q) + \frac{2^{2q} poly(\log n)}{n^{(\log n)^3}}).
    \end{align*}
\end{proof}

\section{Conclusion}
In this paper, we study the power of SAT oracles in the context of approximate model counting and show a lower bound of $\tilde{\Omega}(\log n)$ on the number of oracle calls. This is in contrast to other settings where a SAT oracle is provably more powerful than an NP oracle. In fact, we prove that even with a much more powerful oracle (namely $\sat$ oracle), the number of queries needed to approximately count the number of satisfying assignments of a Boolean formula is $\tilde{\Omega}(\log n)$.

\bibliography{references}

\appendix
\section{Proof of Lemma~\ref{obs:core}}
\label{sec:oblivious-equivalence}

Consider any  general $\counter$, $T$. We will show that there exists a semi-oblivious counter that performs similarly. 
Given a sequence of query-sample pairs $\{(A_1,s_1),\dots,(A_{i-1},s_{i-1})\}$, we say the query $A_i$ is a good strategy by $T$ (given  $\{(A_1,s_1),\dots,(A_{i-1},s_{i-1})\}$) if the counter $T$ can return the correct  output by fixing the next query to $A_i$.
It suffices to show that, given a sequence of query-sample pairs $\{(A_1,s_1),\dots,(A_{i-1},s_{i-1})\}$, if  $A_i$ is a good strategy then any $A'_i$ is also a good strategy if  $A'_{i} \cap \{s_1,\dots,s_{i-1}\} = A_{i} \cap \{s_1,\dots,s_{i-1}\}$ and $|A'_i \cap A| = |A_i \cap A| $   for atoms $A \in At(A_1,\dots,A_{i-1})$. This means that to fix the next query, all it requires to fix the intersection size with each atom $A \in At(A_1,\dots,A_i)$ and a subset of $  \{s_1,\dots,s_{i-1}\}$ (to be included in next query). 
We prove it in the following claim.


\begin{claim}

Suppose $A_{i}$ is a good strategy for  $\{(A_1,s_1),\dots,(A_{i-1},s_{i-1})\}$. Consider $A'_i$ such that $A'_{i} \cap \{s_1,\dots,s_{i-1}\} = A_{i} \cap \{s_1,\dots,s_{i-1}\}$ and $|A'_i \cap A| = |A_i \cap A| $     for atoms $A \in At(A_1,\dots,A_{i-1})$. Then $A'_i$ is also a good strategy for $\{(A_1,s_1),\dots,(A_{i-1},s_{i-1})\}$.

\end{claim}

\begin{proof}

We denote by $\mathcal{S}_{N}$ the symmetric group acting on a set of size $N$.  Any $\sigma\in \mathcal{S}_N$ can be thought of acting on any set of size $N$ (by thinking the elements of the set as numbered $1, \dots, N$ and $\sigma$ acting on the set $[N]$). For any element $x$ in the set, we will denote by $\sigma(x)$ the element after the action of $\sigma$. 
For any $\sigma \in \mathcal{S}_{N}$ and set $A$ (with $|A|=N$)  we denote by $\sigma(A)$ the following set 
$\sigma(A) := \{\sigma(x)\ |\ x \in A\}$.

 Let $\sigma\in \mathcal{S}_{2^n}$ be a permutation acting on the set $\tf^n$. For any $\formula$ observe that $|\satisfying{\formula}| = |\sigma(\satisfying{\formula})|$.
 Since any counter  estimates  $|\satisfying{\formula}|$ only, we observe that if   $A_{i}$ is a good strategy for  $\{(A_1,s_1),\dots,(A_{i-1},s_{i-1})\}$ then  $\sigma(A_{i})$ is also a good strategy for $\{(\sigma(A_1),\sigma(s_1)),\dots,(\sigma(A_{i-1}),\sigma(s_{i-1})\}$ for any $\sigma : \tf^n \rightarrow \tf^n$ that preserves the atoms $\At(A_1, \dots, A_{i-1})$ and the elements $\{s_1, \dots, s_{i-1}\}$.
 
 
 Since   $|A'_i \cap A| = |A_i \cap A| $     for atoms $A \in At(A_1,\dots,A_{i-1})$ and $A'_{i} \cap \{s_1,\dots,s_{i-1}\} = A_{i} \cap \{s_1,\dots,s_{i-1}\}$, there exists a $\sigma$ such that $\sigma(A_j) = A_j$, $\sigma(s_j) = s_j$ for all $j \le i-1$ and also  $\sigma(A_i) = A'_i$. By our earlier observation, $A'_i$ is also a good strategy for $\{(A_1,s_1),\dots,(A_{i-1},s_{i-1})\}$.
\end{proof}






\end{document}